\documentclass[AMA,STIX2COL]{MRM}
\articletype{Technical Note}%

\usepackage{subfig}
\usepackage{amsmath}

\newcommand{\im}{{\rm i}}
\newcommand{\sgn}{{\rm sgn}}

\received{}
\revised{}
\accepted{}
\begin{document}

\title{Sinusoidal Sensitivity Calculation for Line Segment Geometries}

\author[1,2]{Luciano Vinas}{}

\author[2,3]{Atchar Sudyadhom}{}

\authormark{Vinas and Sudyadhom}

\address[1]{\orgdiv{Department of Statistics}, \orgname{UCLA}, \orgaddress{\city{Los Angeles}, \state{California}, \country{USA}}}

\address[2]{\orgdiv{Department of Radiation Oncology}, \orgname{Dana-Farber Cancer Institute | Brigham and Women’s Hospital}, \orgaddress{\city{Boston}, \state{Massachusetts}, \country{USA}}}

\address[3]{\orgname{Harvard Medical School}, \orgaddress{\city{Boston}, \state{Massachusetts}, \country{USA}}}

\corres{Luciano Vinas, Mathematical Sciences, UCLA, 520 Portola Plaza, Los Angeles, CA. \email{lucianovinas@g.ucla.edu}}

\finfo{This work was partially supported by \fundingAgency{NIBIB of the National Institutes of Health} Award number: \fundingNumber{R21EB026086}}

\abstract{\section{Purpose} Provide a closed-form solution to the sinusoidal coil sensitivity model proposed by Kern et al. Solution allows for the precise computations of varied, simulated bias fields which can be directly applied onto raw intensity datasets.
\section{Methods} Fourier distribution theory and standard integration techniques were used to calculate the Fourier transform of measured magnetic field produced from line segment sources.
\section{Results} A $L^1_{\rm loc}(\mathbb{R}^3)$ function is derived in full generality for arbitrary line segment geometries. Sampling criteria and equivalence to the original sinusoidal model are discussed. Lastly a CUDA accelerated implementation \texttt{biasgen} is provided for on-demand sensitivity and bias generation.
\section{Conclusion} Given the modeling flexibility of the simulated procedure, practitioners will now have access to a more diverse ecosystem of simulated datasets which may be used to quantitatively compare prospective debiasing methods.}


\jnlcitation{\cname{%
\author{Vinas L}, and
\author{Sudyadhom A}} (\cyear{2022}), 
\ctitle{Sinusoidal Sensitivity Calculation for Segment Geometries}, \cjournal{Magn Reson Med.}, .}

\maketitle

\footnotetext{\textbf{Abbreviations:}~\hbox{MR,~magnetic~resonance;~INU,~non-uniformity~intensity;} RF, radiofequency; BART, Berkeley Advanced Reconstruction Toolbox}

\section{Introduction}\label{sec:intro}

The success of parallel imaging methods in magnetic resonance (MR) imaging have allowed for quicker image acquisition with little to no cost of spatial aliasing\cite{Deshmane2012}. For parallel imaging methods, reconstruction quality is dependent on how well approximated the gain maps are for the different radiofrequency (RF) receiver coil contributions. These gain normalization maps, also known as coil sensitivity maps, can be difficult to estimate at scan time. When incorrectly normalized, the combined receiver coil contributions can form series of spatial inhomogeneities, known as bias fields, in the final reconstructed image. When left uncorrected, these corruption artifacts can throw off predictive accuracy of different statistical learning models. 

Fortunately, there is a rich literature in image post-processing techniques which aim to amend the MR bias field problem\cite{Tustison10,Wells96}. These techniques may vary from assuming a smooth prior on the generated bias field to choosing a supervised approach at predicting bias fields\cite{Simko22}. In practice these methods are compared against a select few ground-truth datasets, such as the BrainWeb phantom dataset\cite{Cocosco97}. Generally these datasets feature a single fixed bias field, which leaves any supervised method at the risk of overfitting.

\begin{figure}[t]
\centering
    \subfloat[\centering T1-weighted BrainWeb biased data]{{\includegraphics[height=3.5cm]{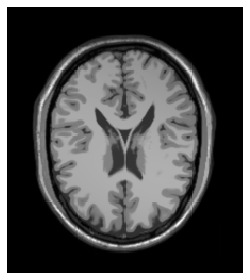} }}%
    \subfloat[\centering T1-weighted thorax MR data]{{\includegraphics[height=3.5cm]{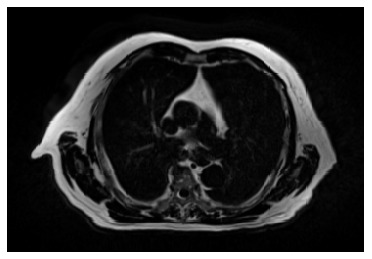} }}%
    \caption{Bias field comparison between 40\% INU BrainWeb and patient thorax MR image. Relative to subfigure (a), subfigure (b) showcases sharp bias intensities at anatomy boundary.}%
    \label{fig:1}%
\end{figure}

Another thing to note is that the bias field featured in the BrainWeb dataset are not very noticeable over the soft tissue areas. As show in Figure \ref{fig:1}, different anatomies which are scanned using different coil configurations may experience more intense bias fields in their final reconstruction. In an ideal setting, our debias testing environment should contain a wide range of physically-viable bias fields with empirically supported functional forms. One prospective model is mentioned in Ref. [\citen{Kern12}] which states that coil sensitivity maps may be well approximated by sparse Fourier representations of the received magnetic field. This is a sensitivity model which is already applied and cited by other popular medical software such as the Berkeley Advanced Reconstruction Toolbox\cite{Uecker15} (BART). 

A current hang-up of the model is that it requires computing the full Fourier transform of some numerically estimated magnetic field. Practical implementations of the model, like the one done by BART, use a small selection of dominant frequencies from an empirically verified source and then apply the fixed sensitivity maps to different phantoms and tissue contrasts. Alternatively one may consider taking fast Fourier transforms (FFT) of some numerical data, but this approach runs into the issue of being resolution limited when considering finer sampling grids.

The goal of this paper will be to provide a closed-form solution to sinusoidal sensitivity model of Ref. [\citen{Kern12}] for a restricted but expressive class of emitted magnetic fields. In particular we are interested in the sensitivity maps produced by magnetic fields generated from compound line segment geometries. The main equation of focus will be sensitivity model
\begin{equation}\label{eq:sens}
S(x) = \sum_{\omega \in G} e^{-\im \omega \bullet x} \mathscr{F}\{\bar{B} (x)\}(\omega),
\end{equation}
where $G\subset\mathbb{R}^3$ is a finite grid of points and 
\begin{equation}\label{eq:meas_B}
\bar{B}(x) =  B(x) \bullet (u + \im v)
\end{equation} 
is the measured magnetic field of some source $B(x)$ according to a readout direction $u\in\mathbb{R}^3$ and a phase encoding direction $v\in\mathbb{R}^3$. Notice that as the sensitivity model (\ref{eq:sens}) and the measured field (\ref{eq:meas_B}) are linear in the source $B$ and $B$ itself follows a superposition property, we can, without loss of generality, consider the Fourier transform of single line segment in order to solve (\ref{eq:sens}) for the class all compound line segment sources $B$.

\section{Line Segment Magnetic Field}\label{sec:mag_field}
\begin{figure}[t]
\centering
\includegraphics[height=5cm]{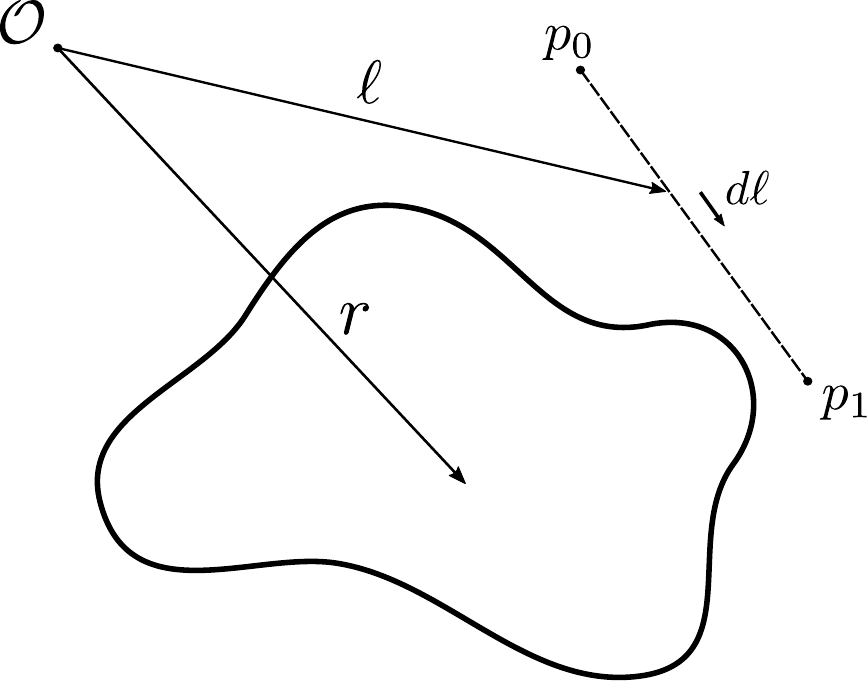}
\caption{Line segment under consideration. Segment $C$ has endpoints specified by $(p_0,p_1)$. Field $B$ is being evaluated at point $r$.}\label{fig:2}
\end{figure}
\noindent We consider the magnetic field contribution of some arbitrary line source $C$ shown in Figure \ref{fig:2}. The contribution of an infinitesimal section of the line segment $d\ell$ can expressed using Biot-Savart law,
$$d B(x)  = \frac{\mu_0 I}{4\pi} \frac{d\ell \times (x-\ell)}{||x-\ell||_2^3}.$$
Assuming a constant current $I$, the total magnetic field is proportional to
$$B(x) \propto \int_{C} \frac{\widehat{d\ell} \times (x-\ell)}{||x-\ell||_2^3},$$
where the right-hand has been normalized by segment length $||p_1 - p_0||_2$. Note that $C$ can be parameterized as 
$$C(t) = p_0 (1-t) + tp_1,\quad\text{for}\; t\in[0,1].$$
A change of variables reveals
$$\widehat{d\ell} = \frac{p_1-p_0}{||p_1 - p_0||_2}\,dt.$$
With shorthands $a = p_1 - p_0$ and $b = x - p_0$ we have
$$B(x) \propto \frac{a\times b}{||a||_2} \int_0^1 \frac{dt}{||b-at||_2^3}.$$
This integral has the following solution
$$\bigg(\frac{a\times b}{||a||_2^2 \,||b||_2^2 - (a\bullet b)^2}\bigg) \frac{-a\bullet (b-at)}{||a||_2 ||b-at||_2}\bigg|_0^1.$$
With the additional shortand $c = b-a$,
$$B(x) \propto \frac{(\hat{a}\times\hat{b})}{||a||_2||b||_2 (1-(\hat{a}\bullet\hat{b})^2)}\big(\hat{a}\bullet \hat{b} - \hat{a}\bullet \hat{c}\big).$$
Using identity $||a\times b||_2^2 = ||a||_2^2 \,||b||_2^2 - (a\bullet b)^2$, we arrive at the final symmetric form
$$B(x) \propto \frac{(\hat{a}\bullet\hat{b})(\hat{a}\times \hat{b})}{||a||_2||b||_2(1-(\hat{a}\bullet\hat{b})^2)}-\frac{(\hat{a}\bullet\hat{c})(\hat{a}\times \hat{c})}{||a||_2||c||_2(1-(\hat{a}\bullet\hat{c})^2)}.$$

\section{Fourier Transform of a Line Magnetic Field}\label{sec:fourier}
Recall the measured magnetic field $\bar{B}(x) = B(x)\bullet (u+\im v)$. Calculating (\ref{eq:sens}) simplifies to calculating the Fourier transform for functions of the form
$$T^{(p_0)}(x) = \frac{(\hat{a}\bullet \hat{b}) \big(u\bullet (\hat{a}\times\hat{b})\big)}{||a||_2||b||_2 (1- (\hat{a}\bullet\hat{b})^2)}.$$
When integrating we will consider the change of variable 
$$x^\prime = H^{\intercal}(x - p_0)$$
where $H$ is an orthonormal rotation matrix. We are interested in rotating to the orthogonal basis 
$$\mathcal{B} = \{\omega - \hat{a}(\omega\bullet \hat{a}),\;\omega \times a ,\;a\},$$
with coordinate representations $a^\prime = (0,0,||a||_2)$ and $\omega^\prime = (\omega_x^\prime,0,\omega_z^\prime)$. In particular we align 
will $a^\prime$ with the $(+z)$-axis and place $\omega^\prime$ in the $(+x,z)$-halfspace. For this reason we will assume $\omega_x^\prime > 0$ throughout without loss of generality. Additionally we note in the case $\omega$ and $a$ are collinear, the Fourier transform $\mathscr{F}\{T\}(\omega) = 0$ regardless of choice of basis. 

This change of variables produces the following simplificaitons
\begin{align*}
||b||_2 &= ||x^\prime||_2,\\
a\bullet b &= (H^\intercal a) \bullet x^\prime,\\
u\bullet (a\times b) &= u \bullet (H H^{\intercal} a)\times (H x^\prime)\\
&= (H^\intercal u)\bullet (H^{\intercal}a)\times x^\prime,
\end{align*}
where $H^\intercal a = (0,0,||a||_2)$.

As the choice of basis $\mathcal{B}$ is independent to input $b$, the $x^\prime$ transformation can be adapted to $T^{(p_1)}$ with $x^{\prime\prime} = H^\intercal(x-p_1)$. Next defining $T = T^{(p_0)} - T^{(p_1)}$ we can combine relevant integrations as so
\begin{align*}
\int T(x) \,e^{\im \omega \bullet x} dx &= \int T^{(p_0)}(x)\,e^{\im \omega \bullet x}\, dx - \int T^{(p_1)}(x)\,e^{\im \omega \bullet x}\, dx\\
&= \big(e^{\im \omega^\prime \bullet H^\intercal p_0} - e^{\im \omega^\prime \bullet H^\intercal p_1}\big) \int T^{(0)}(x^\prime) \,e^{\im \omega^\prime \bullet x^\prime}\, dx^\prime
\end{align*}
For shorthand we introduce
\begin{align*}
A &= e^{\im \omega^\prime \bullet H^\intercal p_0} - e^{\im \omega^\prime \bullet H^\intercal p_1}\\
&= -\im \sin(||a||_2\, \omega_z^\prime/2)\, e^{\im \omega^\prime \bullet H^\intercal (p_1 + p_0)/2},
\end{align*}
where equality follows from $\omega^\prime \bullet H^\intercal (p_1 - p_0) = \omega_z^\prime ||a||_2$. 

Now we begin the task of calculating the Fourier transform of $T$. We will suppress all prime notation as it is understood that integration will be done in the rotated coordinates. Integrating with spherical coordinates
\begin{multline*}
\mathscr{F}\{T\}(\omega) = \frac{A}{||a||_2} \int_0^\infty \int_0^\pi \int_0^{2\pi} \bigg(\frac{\cos\theta}{ r (1-\cos\theta^2)}\bigg)\cdot\\ \big(\sin\theta(u_y \cos\phi - u_x\sin\phi) e^{\rm i r (\omega_x \cos\phi\sin\theta + \omega_z\cos\theta)}\big)\cdot\\(r^2 \sin\theta) \,d\phi\,d\theta\,dr
\end{multline*}
Contribution $u_x\sin\phi$ can be dropped as
$$\int_0^{2\pi} u_x \sin\phi e^{\im r \omega_x \cos\phi \sin\theta} = \frac{\im}{ r\omega_x \sin\theta} e^{\im r\omega_x \cos\phi\cos\theta}\bigg|_{\phi=0}^{2\pi} = 0.$$
Next expand $e^{\im \eta} = \cos \eta + \im \sin \eta$, symmetrize integration bounds, and simplify all odd $(\theta,\phi)$-function contributions
\begin{multline*}
\mathscr{F}\{T\}(\omega) = 2u_y \frac{A}{||a||_2} \int_0^\infty \int_{-\pi/2}^{\pi/2} \int_{-\pi/2}^{\pi/2} (r\sin\theta\sin\phi)\cdot\\ (\sin(r\omega_x \sin\phi\cos\theta)\sin(r\omega_z \sin\theta))\,d\phi\,d\theta\,dr.
\end{multline*}
Although we see now the integrand does not lie in $L^1(\mathbb{R}^3)$, the Fourier transform of $T$ does still exist in a distributional sense. If $T$ is a tempered distribution for the Schwartz class $\mathcal{S}(\mathbb{R}^3)$ then the Fourier transform of $T$ can be defined by $\mathscr{F}\{T\}$ which satisifes
\begin{equation}\label{eq:F_iden}
\int \mathscr{F}\{T\}(\omega)\, \varphi(\omega)\,d\omega = \int T(x) \,\mathscr{F}\{\varphi\}(x)\, dx\hspace{2mm}\forall\varphi\in\mathcal{S}(\mathbb{R}^3).
\end{equation}
Since $T(x) \lesssim ||x||_2^{-1}$, we have by the spherical integration factor of $\mathbb{R}^3$ that $|\int T(x) \varphi(x)\, dx| \leq C_\varphi <\infty$
for every $\varphi\in\mathcal{S}(\mathbb{R}^3)$. It follows that $T_\lambda(x)= T(x) e^{-\lambda ||x||_2}$ can also be identified with a tempered distribution. Lemma \ref{lem:conv} shows
$$\lim_{\lambda\rightarrow 0^+}\int \mathscr{F}\{T_\lambda\}(\omega)\,\varphi(\omega)\,d\omega = \int \mathscr{F}\{T\}(\omega)\,\varphi(\omega) \,d\omega$$
for all $\varphi\in\mathcal{S}(\mathbb{R}^3)$, given that $T \in \mathcal{S}^\prime(\mathbb{R}^3)\cap L^1_{\rm loc}(\mathbb{R}^3)$ is a tempered \textit{function} of polynomial growth. 

Upon confirmation that $\lim_{\lambda\rightarrow 0^+} \mathscr{F}\{T_\lambda\}$ is also a tempered function of polynomial growth, we will have that $(\varphi \lim_{\lambda\rightarrow 0^+} \mathscr{F}\{T_\lambda\})$ is absolutely integrable. The immediate consequence from Lebesgue dominated convergence theorem is
$$\int \mathscr{F}\{T\}(\omega)\,\varphi(\omega) \,d\omega = \int \Big(\lim_{\lambda \rightarrow 0^+}\mathscr{F}\{T_\lambda\}(\omega)\Big)\varphi(\omega)\,d\omega.$$
Since $T_\lambda \in L^1(\mathbb{R}^3)$, we may directly evaluate $\lim_{\lambda\rightarrow 0^+}\mathscr{F}\{T_\lambda\}$ while maintaining equality to $\mathscr{F}\{T\}$ without needing to rely on identity (\ref{eq:F_iden}). 

With this in mind, we suppress the limit notation of our left-hand side equalties and carry-on with the integration
\begin{multline*}
\mathscr{F}\{T\}(\omega) = \lim_{\lambda\rightarrow 0^+} 2u_y \frac{A}{||a||_2}\int_0^\infty\int_{-\pi/2}^{\pi/2}\int_{-\pi/2}^{\pi/2}(r\sin \theta \sin\phi)\cdot\\
\big(\sin(r\omega_x\sin\phi \cos\theta)\sin(r\omega_z\sin\theta)e^{-\lambda r}\big)\,d\phi\,d\theta\,dr.
\end{multline*}
Apply trigonometric identity 
$$\sin(x)\sin(y) = (\cos(x-y) - \cos(x+y))/2,$$
and integrate with respect $r$,
\begin{multline*}
\mathscr{F}\{T\}(\omega) = \lim_{\lambda\rightarrow 0^+}-u_y\frac{ A}{||a||_2}\int_{-\pi/2}^{\pi/2}\int_{-\pi/2}^{\pi/2}(\sin \theta \sin\phi)\cdot\\\bigg(\frac{\lambda^2 - (\omega_x\sin\phi\cos\theta + \omega_z\sin\theta)^2}{(\lambda^2 + (\omega_x\sin\phi\cos\theta + \omega_z\sin\theta)^2)^2}-\\ \frac{\lambda^2 - (\omega_x\sin\phi\cos\theta - \omega_z\sin\theta)^2}{(\lambda^2 + (\omega_x\sin\phi\cos\theta - \omega_z\sin\theta)^2)^2}\bigg) \,d\phi\,d\theta.
\end{multline*}
As shown by corollary \ref{cor:domin}, this integrand can be proportionally dominated by the $\lambda = 0$ evaluation. By dominated convergence theorem
\begin{multline*}
\mathscr{F}\{T\}(\omega) = -4u_y\frac{A}{||a||_2}\int_{-\pi/2}^{\pi/2}\int_{-\pi/2}^{\pi/2}\\ \Bigg(\frac{\omega_x\omega_z\sin^2\theta\cos\theta\sin^2\phi}{(\omega_x^2\sin^2\phi\cos^2\theta - \omega_z^2\sin^2\theta)^2}\Bigg)\,d\phi\,d\theta.
\end{multline*}
Consider the change of variables $t=\sin\theta$ and integrate with respect to $t$ using a partial fraction decomposition,
\begin{multline*}
\mathscr{F}\{T\}(\omega) = -8u_y\omega_z\frac{A}{\omega_x||a||_2}\int_{0}^{\pi/2}\bigg(\frac{1}{1 + (\omega_z^2/\omega_x^2)\csc^2\phi}\bigg)\cdot \\\Bigg(\frac{1}{\omega_z^2}
+ \frac{1}{\omega_x^2}\frac{\csc^2\phi\tanh^{-1}(1+(\omega_z^2/\omega_x^2)\csc^2\phi)^{-1/2}}{\sqrt{1 + (\omega_z^2/\omega_x^2)\csc^2\phi}} \Bigg)\, d\phi
\end{multline*}
For brevity we will consider the shorthands
$$\mathscr{F}\{T\}(\omega) = -8u_y\omega_z\frac{A}{\omega_x||a||_2}\bigg(\frac{I_1}{\omega_z^2} + \frac{I_2}{\omega_x^2}\bigg).$$
Calculating first integration term $I_1$
$$\frac{I_1}{\omega_z^2} = \frac{1}{\omega_z^2}\int_{0}^{\pi/2}\frac{1}{1 + (\omega_z^2/\omega_x^2)\csc^2\phi}\,d\phi = \frac{\pi}{2}\frac{1}{\omega_z^2} \bigg(1 - \frac{|\omega_z|}{||\omega||_2}\bigg).$$
This equality follows from $\lim_{x\rightarrow \pi/2} c\tan^{-1}(\tan (x)/c) = |c|$ and $\omega_x > 0$. Next with trigonometric identity
$$\csc^2(\phi) = 1 + \cot^2(\phi),$$
introduce change of variables $t= \cot(\phi)$ and shorthands $\rho^2= 1 + (\omega_z^2/\omega_x^2)$ and $\kappa^2 = \rho^2 - 1$ to simplify the second integration term as
$$I_2 = \int_{0}^\infty \frac{1}{(\rho^2 +  \kappa^2 t^2)^{3/2}}\tanh^{-1}\frac{1}{\sqrt{\rho^2 + \kappa^2t^2}}\, dt$$
Calculating $I_2$
\begin{align*}
I_2 &= 0+ \frac{1}{\rho^2}\int_0^\infty \frac{t^2\,dt}{(1+t^2)(\rho^2 + \kappa^2 t^2)}\\
&= \frac{\pi}{2}\bigg(\frac{1}{\rho |\kappa|} - \frac{1}{\rho^2}\bigg).
\end{align*}
Here a partial fraction decomposition was used on the integration-by-parts integral.Written in terms of $\omega$,
$$\frac{I_2}{\omega_x^2} = \frac{\pi}{2}\frac{1}{\omega_z^2} \Big(\frac{|\omega_z|}{||\omega||_2} - \frac{\omega^2_z}{||\omega||_2^2}\Big).$$
Summed together with equality $\omega_x^2 = ||\omega||_2^2- \omega_z^2$,
$$\mathscr{F}\{T\}(\omega) = -4\pi u_y\frac{A}{||a||_2}\frac{\omega_x}{\omega_z} \frac{1}{||\omega||_2^2}.$$
Lastly expanding $A$ in terms of the original, canonical coordinates
\begin{multline}\label{eq:T_transf}
\mathscr{F}\{T\}(\omega) = \frac{4\pi \im (u \bullet \widehat{\omega\times a})}{||\omega||_2}\,e^{{\rm i} \omega\bullet \frac{p_1 + p_0}{2}}\,\text{sinc}\bigg(\frac{\omega \bullet a}{2}\bigg)\cdot\\ \sqrt{1 - (\hat{\omega}\bullet \hat{a})^2}.
\end{multline}
Note that as
$$\big|\lim_{\lambda \rightarrow 0^+}\mathscr{F}\{T_\lambda\}(\omega)\big|\lesssim ||\omega||_2^{-1}$$
we confirm our earlier claim that the pointwise limit of $\mathscr{F}\{T_\lambda\}$ is a tempered function of polynomial growth in $\mathbb{R}^3$.

When combined with the $v$ measurement contribution, we obtain the final closed-form solution to the sinusoidal sensitivity model. Examples of simulated sensitivities and bias fields for different grids $G$ can be found in Appendix \ref{app:simul}.

\section{Sampling with the Sinusoidal Sensitivity Model}\label{sec:sampling}
Consider the sparse Fourier sampler
$$J(\eta) = \sum_{\omega\in G}\delta(\eta-\omega)$$
where $\delta(\cdot-\omega)$ is the Dirac delta distribution centered at $\omega$. Understood formally, the sinusoidal sensitivity model can be expressed in terms of this sampler as
$$S(x) = \mathscr{F}^{-1}\{J \mathscr{F}\{\bar{B}\}\}(x).$$
An issue found with sampler $J$ is that its action on function $f$ is only well-defined if the function $f$ has meaningful point evalutions. The framework used in Section \ref{sec:fourier} worked with class $\mathscr{F}\{\bar{B}\}\in L^1_{\rm loc}(\mathbb{R}^3)$ whose behavior is only specified with respect to an integrating action against the Lebesgue measure on $\mathbb{R}^3$.

In hopes to extend the sinusoidal model to work better with the equivalence class $\mathscr{F}\{\bar{B}\}$ we may consider the generalized sampler
$$J_R = \sum_{\omega\in G}\frac{1}{{\rm vol}(B(\omega,R))}\chi_R(\eta;\omega)$$
where $\chi_R(\eta;\omega) = 1\{||\eta-\omega||_2 < R\}$ is the $R$-cutoff function centered at $\omega$ and ${\rm vol}(B(\omega,R))$ is the volume of the $R$-radius sphere centered at $\omega$. We see then for $R>0$, the generalized sinusoidal model 
$$S_R(x) = \mathscr{F}^{-1}\{J_R\, \mathscr{F}\{\bar{B}\}\}(x)$$
does produce a well-defined result. Ideally we would like limit $\lim_{R\rightarrow 0^+} S_R$ to be equivalent to our original sinusoidal sensitivity model. 

Before continuing, we introduce the following term for notational clarity. With $q = u+\im v$ consider
\begin{multline*}
g(\omega) = \frac{4\pi \im (q \bullet \widehat{\omega\times a})}{||\omega||_2}\,\text{sinc}\bigg(\frac{\omega \bullet a}{2}\bigg) \sqrt{1 - (\hat{\omega}\bullet \hat{a})^2}\,\frac{e^{{\rm i}\omega\bullet \frac{p_1 + p_0}{2}}}{e^{{\rm i}\omega\bullet x}}
\end{multline*}
which lies in the space of continuous functions. That is have $g(\omega)\in C(\mathbb{R}^3\setminus\{0\})$ and $e^{-\im\omega\bullet x}\,\mathscr{F}\{\bar{B}\}(\omega) \in L^{1}_{\rm loc}(\mathbb{R}^3)$. One can show that these objects satisfy the limit equality
\begin{align*}
g(\omega) = \lim_{R\rightarrow 0^+}\frac{1}{{\rm vol}(B(\omega,R))}\int_{||\eta-\omega||_2 < R} e^{-\im\eta\cdot x}\,\mathscr{F}\{\bar{B}\}(\eta)\,d\eta.
\end{align*}
To see this note that, for all continuity points of $g$, one can pair any deviation $\varepsilon > 0$ with a suitably small radius $R$ such that the following implication holds
$$|e^{-\im\eta\bullet x}\,\mathscr{F}\{\bar{B}\}(\eta) - g(\omega)|<\varepsilon$$
for almost every $\eta \in B(\omega, R)$. With shorthand $\langle f,g\rangle_\eta = \int f(\eta)\,g(\eta)\, d\eta$ and $\omega$-centered indicator $\chi_R(\eta;\omega) = 1\{||\eta - \omega||_2 < R\}$, it follows that
\begin{align*}
\big|\langle e^{-\im\eta\bullet x}\,\mathscr{F}\{\bar{B}\},\chi_R \rangle_\eta - g(\omega)\big| &\leq \big\langle |e^{-\im\eta\bullet x}\,\mathscr{F}\{\bar{B}\} -g(\omega)|,\chi_R\big\rangle_\eta\\
&< \langle \varepsilon, \chi_R\rangle_\eta\\
&= {\rm vol}(B(\omega,R)) \cdot \varepsilon. 
\end{align*}
As a consequence the equality
$$S(x) = \lim_{R\rightarrow 0^+} S_R(x)\quad\text{for all}\;x\in\mathbb{R}^3,$$
is well-defined with $S = \mathscr{F}^{-1}\{J g\}$ and any grid sampler $J$ which does not contain the origin.

\section{Implementation}
A CUDA accelerated implementation of the sinusoidal model \texttt{biasgen} can be found on authors' GitHub\footnote{\href{https://github.com/lucianoAvinas/biasgen}{https://github.com/lucianoAvinas/biasgen}}. This software is a Python package which takes in user-defined coil positions and sampling information to produce fully custom, 3-dimensional bias fields. A coil positioning and visualization tool is provided to help with setup. Various examples of how to run \texttt{biasgen} can be found in the examples subdirectory of the repository. 

Provided in Figures \ref{fig:sm}and \ref{fig:rg}are simulated bias fields superimposed on T1-weighted BrainWeb data. The bias fields $\bar{S}(r)$ were created by combining sensitivity maps $S_k(r)$ through a sum of squares approach
$$\bar{S}(r) = \bigg(\sum_{k=1}^N S_k^2(r)\bigg)^{1/2}.$$
Relevant parameters for forming sensitivity maps $S_k(r)$ were the sampling grid half-length $L$, sampling spacing factor $s= (s_x,s_y,s_z)$, and the sampling starting shift $c = (c_x, c_y, c_z)$. A per-coil breakdown of the individual emitter/receiver contributions for Figures \ref{fig:sm}and \ref{fig:rg}can be found in Appendix \ref{app:simul}.

\begin{figure}[t]%
    \centering
    \subfloat[\centering Top-down coil view]{{\includegraphics[height=4cm]{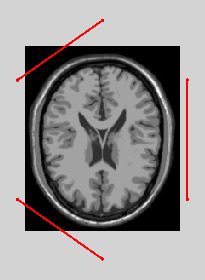}}}%
    \hspace{2mm}
    \subfloat[\centering BrainWeb T1 biased data]{{\includegraphics[height=4cm]{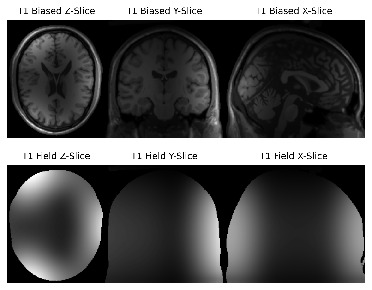}}}%
    \caption{Smooth bias field simulation with top-down view of the coil geometry (left). Sensitivity map parameters: $L=7$, $s=(1,1,1)$, $c = (0,0, 0)$.}\label{fig:sm}
\end{figure}

\begin{figure}[h]%
    \centering
    \subfloat[\centering Top-down coil view]{{\includegraphics[height=4cm]{coil_geom_topdown}}}%
    \hspace{2mm}
    \subfloat[\centering BrainWeb T1 biased data]{{\includegraphics[height=4cm]{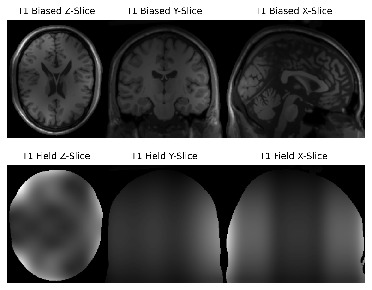}}}%
    \caption{Rough bias field simulation with top-down view of the coil geometry (left). Sensitivity map parameters: $L=12$, $s=(1.75,1.75,0.5)$, $c = (-0.25,-0.25, 0)$.}\label{fig:rg}
\end{figure}

In practice many realistic geometries can be recreated or approximated using series of line segments. For example smooth curve geometries such as circles may be approximated using sufficiently high degree polygons instead. The line segments of these polygons can then be modeled using equation (\ref{eq:T_transf}) from before. 

Lastly as the sensitivity maps of (\ref{eq:T_transf}) can be produced independently of underlying image intensity, \texttt{biasgen} can be a useful tool in augmenting existing MR dataset with more pronounced bias fields. Realistic data augmentation have been shown to increase the performance of certain neural models like the convolutional neural network\cite{Simard03}. This effect may be more pronounced in the case the supervised learning task is already data-limited.

\section{Conclusion}\label{sec:conc}

In this paper, we have derived in closed-form the Fourier transform of measured magnetic fields emitted by line segment geometries. This equation can be used to simulate realistic coil sensitivity maps up to arbitrary accuracy and sensitivity smoothness. Special care was taken to discuss the distributional nature of the solved Fourier transform and settings were identified where this closed-form agreed with the sparse sampled model introduced by Ref. [\citen{Kern12}]. A software package \texttt{biasgen} used for implementing the various equations of the paper has been provided by the authors. As next steps, further work can be done to solve (\ref{eq:sens}) for smooth line curves, such as for the case of circular or cylindrical geometries.

\section*{Acknowledgments}
Research reported in this manuscript was partially supported by the NIBIB of the National Institutes of Health under award number R21EB026086. The content is solely the responsibility of the authors and does not necessarily represent the official views of the National Institutes of Health.

\subsection*{Conflict of interest}

The authors declare no potential conflict of interests.

\bibliography{sens_calc}%

\begin{thebibliography}{1}

\bibitem{Deshmane2012}
Deshmane A, Gulani V, Griswold MA, Seiberlich N. Parallel MR imaging.  {\it
  Journal of magnetic resonance imaging : JMRI. }2012;36(1):55-72.

\bibitem{Tustison10}
Tustison NJ, Avants BB, Cook PA, et al. N4ITK: Improved N3 Bias Correction.
  {\it IEEE Transactions on Medical Imaging. }2010;29(6):1310-1320.

\bibitem{Wells96}
Wells WM, Grimson WEL, Kikinis R, Jolesz FA. Adaptive segmentation of MRI data.
   {\it IEEE Transactions on Medical Imaging. }1996;15(4):429-442.

\bibitem{Simko22}
Simko AT, L{\"o}fstedt T, Garpebring A, Nyholm T, Jonsson J. {MRI} bias field
  correction with an implicitly trained {CNN}.  In:  {\it Medical Imaging with
  Deep Learning}; 2022.

\bibitem{Cocosco97}
Cocosco CA, Kollokian V, Kwan RKS, Evans AC. BrainWeb: Online Interface to a 3D
  MRI Simulated Brain Database.  {\it NeuroImage. }1997;5(4).

\bibitem{Kern12}
Guerquin-Kern M, Lejeune L, Pruessmann KP, Unser M. Realistic Analytical
  Phantoms for Parallel Magnetic Resonance Imaging.  {\it IEEE Transactions on
  Medical Imaging. }2012;31(3):626-636.

\bibitem{Uecker15}
Uecker M, Ong F, Tamir JI, et al. Berkeley advanced reconstruction toolbox.
  {\it Proc. Intl. Soc. Mag. Reson. Med. }2015;23(2486).

\bibitem{Simard03}
Simard PY, Steinkraus D, Platt JC. Best practices for convolutional neural
  networks applied to visual document analysis.  In:  {\it Seventh
  International Conference on Document Analysis and Recognition, 2003.
  Proceedings.}:958-963; 2003.

\end{thebibliography}
\vfill\pagebreak


\vspace*{6pt}
\appendix

\section{Auxiliary Lemmas}
\subsection{A Specific Tempered Function Convergence}\label{app:conv}
\begin{definition}[Polynomial Growth Tempered Functions]\label{def:1}
We say locally integrable $T \in L^1_{\rm loc}(\mathbb{R}^n)$ has polynomial growth $m > 0$ if there exists some constant $C$ such that for all $R\geq 1$,
$$\int_{||x||_2 \leq R} |T(x)|\,dx \leq C R^{m}.$$
\end{definition}
As the action $\langle T,g\rangle = \int T(x) g(x)\,dx$ with any locally integrable $T$ defines a tempered distribution for $g\in\mathcal{S}(\mathbb{R}^n)$, we will sometimes make the distinction between the tempered function $T$ and the tempered distribution $T$. A consequence of definition \ref{def:1}, is that for every polynomial-growing tempered function $T$ there is some $R^\prime$ such that $T(x)$ is dominated as
$$|T(x)| \leq C^\prime ||x||^{m-n-2}_2,$$
for $||x|| \geq R^\prime$ and some constant $C^\prime > 0$.

\begin{lemma}\label{lem:conv}
Let $T$ be a tempered function of polynomial growth $m_0 > 0$ and define $T_\lambda(x)=T(x) e^{-\lambda ||x||_2}$. Then for every $\varphi\in\mathcal{S}(\mathbb{R}^n)$ we have
$$\lim_{\lambda \rightarrow 0^+}\int \mathscr{F}\{T_\lambda\}(\omega)\, \varphi(\omega)\,d\omega = \int \mathscr{F}\{T\}(\omega)\, \varphi(\omega)\,d\omega.$$ 
\end{lemma} 
\begin{proof}
Note by identity (\ref{eq:F_iden}) and the fact $\mathscr{F}:\mathcal{S}(\mathbb{R}^n)\rightarrow\mathcal{S}(\mathbb{R}^n)$ is bijective, it suffices to show
$$\lim_{\lambda \rightarrow 0^+}\int T_\lambda(x) \varphi(x)\,dx = \int T(x) \varphi(x)\,dx\quad\forall\varphi\in\mathcal{S}(\mathbb{R}^n).$$ 
Introduce shorthands $\chi_R(x) = 1\{||x||_2\leq R\}$ and $\langle \phi,\varphi\rangle = \int \phi(x) \varphi(x)\,dx$ for every $(\phi,\varphi)\in \mathcal{S}^\prime(\mathbb{R}^n)\times\mathcal{S}(\mathbb{R}^n)$. It follows that
\begin{align*}
|\langle T_\lambda, \varphi\rangle - \langle T, \varphi\rangle| &\leq |\langle T \chi_{R} - T, \varphi\rangle| + |\langle T_\lambda\chi_{R} - T\chi_{R}, \varphi\rangle| +\\ &\hspace{5mm} |\langle T_\lambda- T_\lambda \chi_{R}, \varphi\rangle|\\
&\leq 2|\langle T \chi_{R} - T, \varphi\rangle | + |\langle T_\lambda\chi_{R}-T\chi_{R}, \varphi\rangle|.
\end{align*}
Given the rapid-decay of Schwartz functions $\varphi\in\mathcal{S}(\mathbb{R}^n)$ there exists sufficiently large $R^{\prime\prime}$ such that
$$\bigg|\int_{||x||_2 > R} T(x) \varphi(x)\, dx\bigg| \leq \bigg|\int_{||x||_2 > R} T(x) \,||x||^{-m_0}\,dx\bigg|,$$
for $R \geq R^{\prime\prime}$. The growth conditions on $T$ imply that this upperbound goes to 0 as $R\rightarrow\infty$. That is, for every $\varepsilon>0$ there is some $R_0$ such that
$$R \geq R_0 \implies |\langle T \chi_{R}, \varphi\rangle - \langle T, \varphi\rangle| < \varepsilon/4.$$
Furthermore note that this relationship is monotonic such that a decrease in $\varepsilon$ produces a non-strict increase in $R_0$. Using the integrability of $T\chi_R \in L^1(\mathbb{R}^n)$
\begin{align*}
|\langle T_\lambda\chi_{R}, \varphi\rangle - \langle T\chi_{R}, \varphi\rangle| &\leq ||T\chi_{R}  \varphi||_{L^1} \big|\big|\chi_{R}\big(1 - e^{\lambda ||x||_2}\big)\big|\big|_{L^\infty}\\
&\leq (1 - e^{-\lambda R}) \sup_{x\in\mathbb{R}^n}|\varphi(x)|\, ||T\chi_{R}||_{L^1}\\
&\leq (1 - e^{-\lambda R}) \,c_\varphi C_R
\end{align*}
where $c_{\varphi},C_R < \infty$ by assumption on $\varphi$ and $T\chi_R$. Collect both constants into one constant $C_{\varphi,R}$. If $C_{\varphi,R_0} \leq \varepsilon/2$ then directly plugging in gives
$$|\langle T_{\lambda},\varphi\rangle - \langle T,\varphi\rangle|<\varepsilon\quad\text{for all}\;\lambda >0.$$
For the complementary condition $C_{\varphi,R_0} > \varepsilon/2$ we have
$$\lambda < -\log (1 - \varepsilon/(2C_{\varphi,R_0}))/R_0\implies |\langle T_{\lambda},\varphi\rangle - \langle T,\varphi\rangle|<\varepsilon.$$
As this can be done for any $\varphi\in\mathcal{S}(\mathbb{R}^n)$, we arrive at the desired relation
$$\lim_{\lambda\rightarrow 0^+}\int \mathscr{F}\{T_\lambda\}(\omega)\,\varphi(\omega)\,d\omega = \int \mathscr{F}\{T\}(\omega)\,\varphi(\omega) \,d\omega.$$
\end{proof}

\subsection{Domination of a Specific Rational Function}\label{app:domin}
\begin{lemma}\label{lem:domin}
For any $a,b\in\mathbb{R}$ we have
$$\bigg|\frac{\lambda^2 - (a+b)^2}{(\lambda^2 + (a+b)^2)^2}-\frac{\lambda^2 - (a-b)^2}{(\lambda^2 + (a-b)^2)^2}\bigg| \lesssim \frac{1}{(a-b)^2} - \frac{1}{(a+b)^2}.$$
for all $\lambda \in \mathbb{R}$.
\end{lemma}
\begin{proof}
Let $A$ be a placeholder value for the absolute value term in the lemma description. With some manipulations
\begin{align*}
A &= 4|ab|\bigg|\frac{3\lambda^4 + 2\lambda^2(a^2+b^2) -(a^2-b^2)^2}{(\lambda^2 + (a+b)^2)^2(\lambda^2 + (a-b)^2)^2}\bigg|\\
 &\leq 4|ab|\bigg|\frac{3\lambda^4 + 6\lambda^2(a^2+b^2)+3(a^2-b^2)^2}{(\lambda^2 + (a+b)^2)^2(\lambda^2 + (a-b)^2)^2}\bigg|\\
 &= 3\bigg|\frac{4ab}{(\lambda^2 + (a+b)^2)(\lambda^2 + (a-b)^2)}\bigg|\\
&\leq 3 \bigg|\frac{4ab}{(a^2-b^2)^2}\bigg|\\
&= 3\,\text{sgn}(ab)\bigg(\frac{1}{(a-b)^2} - \frac{1}{(a+b)^2}\bigg)
\end{align*}
\end{proof}

Define a $\lambda$-parameterized function $f:D\rightarrow \mathbb{R}$ where
$$f(\eta;\lambda) = c\bigg(\frac{\lambda^2 - (a+b)^2}{(\lambda^2 + (a+b)^2)^2}-\frac{\lambda^2 - (a-b)^2}{(\lambda^2 + (a-b)^2)^2}\bigg)$$
and $a,b,c$ are all functions of $\eta$.

\begin{corollary}{\label{cor:domin}}
Suppose functions $a,b,c$ satisfy
$$\sgn(a(\eta)b(\eta)c(\eta)) = \sgn(a(\eta^\prime)b(\eta^\prime)c(\eta^\prime))$$
for all $\eta,\eta^\prime \in D$. Then for integrable $f(\eta;0)$ we have
$$\lim_{\lambda\rightarrow 0^+}\int f(\eta;\lambda)\,d\eta = \int f(\eta;0)\,d\eta.$$
\end{corollary}
\begin{proof}
Note a simple extension of Lemma \ref{lem:domin} shows
\begin{multline*}
\bigg|c\bigg(\frac{\lambda^2 - (a+b)^2}{(\lambda^2 + (a+b)^2)^2}-\frac{\lambda^2 - (a-b)^2}{(\lambda^2 + (a-b)^2)^2}\bigg)\bigg| \leq 3c\,\text{sgn}(abc)\cdot \\ \bigg(\frac{1}{(a-b)^2} - \frac{1}{(a+b)^2}\bigg).
\end{multline*}
As $\text{sgn}(abc)$ is fixed for every $\eta\in D$, we have that $3\,\sgn(abc) f(\eta;0)$ dominates $|f(\eta;\lambda)|$ for all $\lambda \in \mathbb{R}_{>0}$.
\end{proof}

\section{Simulated Sensitivity}\label{app:simul}

\subsection{Smooth Three Coil Map}

Sensitivity maps were generated using three rectangular coils and sampling settings $L = 7$, $s = (1.5,1.5,1.5)$, $c = (0,0,0)$.

\begin{figure}[h]%
    \centering
    \includegraphics[height=6.5cm]{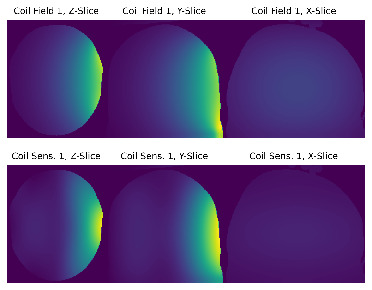}
    \caption{Emitted field and sensitivity comparison for coil 1.}%
\end{figure}

\begin{figure}%
    \centering
    \includegraphics[height=6.5cm]{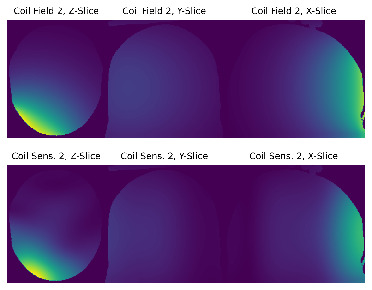} 
    \caption{Emitted field and sensitivity comparison for coil 2.}%
\end{figure}

\begin{figure}%
    \centering
    \includegraphics[height=6.5cm]{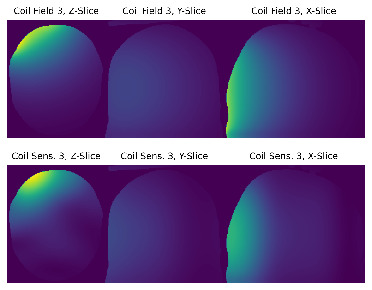}
    \caption{Emitted field and sensitivity comparison for coil 3.}%
\end{figure}

\newpage

\subsection{Rough Three Coil Map}
Sensitivity maps were generated using three rectangular coils and sampling settings $L = 12$, $s = (1.75,1.75,0.5)$, $c = (-0.25,-0.25,0)$.

\begin{figure}[h]%
    \centering
    \includegraphics[height=6.5cm]{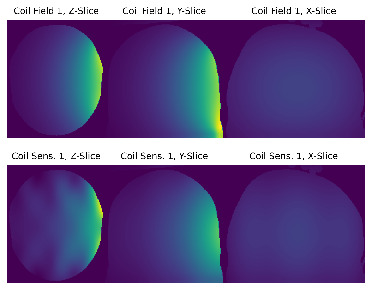}
    \caption{Emitted field and sensitivity comparison for coil 1.}%
\end{figure}

\begin{figure}%
    \centering
    \includegraphics[height=6.5cm]{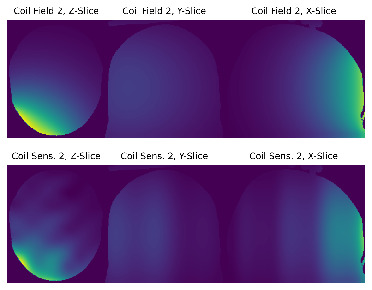}
    \caption{Emitted field and sensitivity comparison for coil 2.}%
\end{figure}

\begin{figure}%
    \centering
    \includegraphics[height=6.5cm]{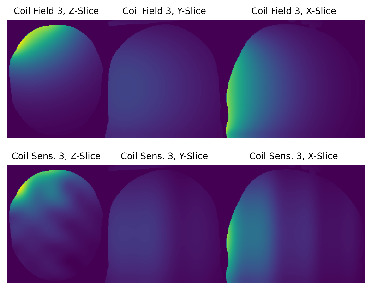}
    \caption{Emitted field and sensitivity comparison for coil 3.}%
\end{figure}

\end{document}